\documentclass[10pt,journal]{IEEEtran}

\usepackage{amssymb}
\usepackage{amsmath}
\usepackage{cite}
\usepackage{url}
\usepackage{xcolor}
\usepackage{cite,graphicx,amsmath,amssymb}
\usepackage{fancyhdr}
\usepackage{mdwmath}
\usepackage{mdwtab}
\usepackage{caption}
\usepackage{amsthm}
\usepackage{setspace}
\usepackage{algpseudocode}
\usepackage{algorithm}
\usepackage{float}
\usepackage[labelformat=simple]{subcaption}
\usepackage{hyperref}
\hypersetup{colorlinks=false}

\graphicspath{{./img/}}

\newtheorem{remark}{Remark}
\newtheorem{theorem}{Theorem}

\newtheorem{lemma}{Lemma}

\newtheorem{corollary}{Corollary}

\newtheorem{proposition}{Proposition}

\allowdisplaybreaks

\algnewcommand{\IfThenElse}[3]{
  \State \algorithmicif\ #1\ \algorithmicthen\ #2\ \algorithmicelse\ #3}

\begin{document}

\title{Coupled Phase-Shift STAR-RISs: A General Optimization Framework}

\author{
        Zhaolin~Wang,~\IEEEmembership{Graduate Student Member,~IEEE,}
        Xidong~Mu,~\IEEEmembership{Member,~IEEE,} \\
        Yuanwei~Liu,~\IEEEmembership{Senior Member,~IEEE,}
        and Robert Schober,~\IEEEmembership{Fellow,~IEEE}
        
\thanks{This work was supported in part by the PHC Alliance Franco-British Joint Research
Programme under Grant 822326028; and in part by the Engineering and Physical Sciences Research Council under Project EP/W035588/1.}
\thanks{Zhaolin Wang, Xidong Mu, and Yuanwei Liu are with the School of Electronic Engineering
and Computer Science, Queen Mary University of London, London E1 4NS, U.K. (e-mail: zhaolin.wang@qmul.ac.uk, xidong.mu@qmul.ac.uk, yuanwei.liu@qmul.ac.uk).}
\thanks{Robert Schober is with the Institute for Digital Communications, Friedrich-Alexander-Universität Erlangen-Nürnberg (FAU), 91054 Erlangen, Germany (e-mail: robert.schober@fau.de).}
}

\maketitle

\begin{abstract}
  A general optimization framework is proposed for simultaneously transmitting and reflecting reconfigurable intelligent surfaces (STAR-RISs) with coupled phase shifts, which converges to the Karush– Kuhn–Tucker (KKT) optimal solution under some mild conditions. More particularly, the amplitude and phase-shift coefficients of STAR-RISs are optimized alternately in closed form. To demonstrate the effectiveness of the proposed optimization framework, the throughput maximization problem is considered in a case study. It is rigorously proved that the KKT optimal solution is obtained. Numerical results confirm the effectiveness of the proposed optimization framework compared to baseline schemes.
\end{abstract}


\section{Introduction}
Recently, the novel concept of simultaneously transmitting and reflecting reconfigurable intelligent surfaces (STAR-RISs) has been proposed \cite{liu2021star, mu2021simultaneously}. In contrast to the conventional reflecting-only RISs that require the transmitter and receiver to be located on the same side \cite{huang2019reconfigurable}, STAR-RISs can transmit and reflect the incident signals to both sides, thus enabling a \emph{full-space} smart radio environment. 

Due to this unique benefit, STAR-RISs have attracted significant attention. However, most of the existing works on STAR-RISs assumed that the phase shifts of the transmission and reflection coefficients can be independently adjusted, which requires complex RIS hardware. Recently, the authors of \cite{xu2022star} have shown that for low-cost passive lossless STAR-RISs, the phase-shift coefficients for transmission and reflection are \emph{coupled} with each other, which introduces new design challenges. To tackle the coupled phase shifts, an element-wise optimization algorithm was proposed in \cite{liu2021simultaneously} for a system comprising a single-antenna transmitter and two users. For the weighted-sum-rate-maximization problem, an alternating optimization (AO) algorithm, which exhaustively searches sets of discrete amplitudes and phase shifts, and thus entails high complexity, was proposed in \cite{niu2022weighted}. Noteworthy, for the aforementioned algorithms, the optimality of the obtained solution cannot be guaranteed. This motivates us to propose a general optimization framework with provable optimality for STAR-RISs with coupled phase shifts.

In this article, we propose a general penalty-based optimization framework for STAR-RISs with coupled phase shifts, where the amplitude and phase-shift coefficients of the STAR-RISs are alternately updated in closed form exploiting the existing penalty dual decomposition (PDD) framework \cite{shi2020penalty}. We prove that the Karush–Kuhn–Tucker (KKT) optimal solution is obtained under some mild conditions. Next, as a case study, we consider throughput maximization in STAR-RIS-aided wireless communication systems to verify the effectiveness of the proposed algorithm. Our numerical results reveal that the throughput achieved with coupled phase-shift STAR-RISs is close to that of more complex independent phase-shift STAR-RISs.

\section{A Coupled Phase-Shift Model for STAR-RISs}

Let $s_n$ denote the incident signal for the $n$-th element of the considered $N$-element STAR-RIS, where $n \in \mathcal{N} = \{1,...,N\}$. Then, the corresponding transmitted signal $t_n$ and reflected signal $r_n$ are given by $t_n = \beta_{t,n} e^{j \phi_{t,n}} s_n$ and $ r_n = \beta_{r,n} e^{j \phi_{r,n}} s_n$, respectively, where $\beta_{t,n} \in [0,1] $ and $\beta_{r,n} \in [0,1]$ denote the real-valued transmission and reflection amplitudes, and $\phi_{t,n} \in [0, 2\pi)$ and $\phi_{r,n} \in [0, 2\pi)$ denote the corresponding phase shifts \cite{liu2021star}. In practice, the values of the amplitudes and phase shifts are determined by the corresponding electric and magnetic impedances of the STAR-RIS. According to the analysis in \cite{xu2022star}, passive lossless STAR-RISs have to meet the following two constraints:
\begin{subequations}
  \begin{align}
    &\beta_{t,n}^2 + \beta_{r,n}^2 = 1, \\
    \label{eqn:coupled}
    &\cos(\phi_{t,n} -  \phi_{r,n}) = 0.
  \end{align} 
\end{subequations}
Specifically, the first constraint stems from the law of energy conservation, and the second constraint referred to as the \emph{coupled phase-shift constraint} is due to the zero-valued real part of the electric and magnetic impedances of the passive STAR-RIS elements. Note that the coupled phase-shift constraint is a non-convex nonlinear equality constraint implying that the absolute phase-shift difference $|\phi_{t,n} -  \phi_{r,n}|$ can only assume values $\frac{\pi}{2}$ and $\frac{3\pi}{2}$. Using existing methods, it is challenging to transform \eqref{eqn:coupled} into a convex form. To overcome this obstacle, in the following section, a general optimization framework is proposed to handle the coupled phase-shift constraint efficiently.

\section{A General Optimization Framework for STAR-RISs with Coupled Phase Shifts}
Consider the following optimization problem:
\begin{subequations} \label{problem:original}
  \begin{align}
      \min_{ \mathbf{x} \in \mathcal{X}, \boldsymbol{\theta}_t, \boldsymbol{\theta}_r} & F(\mathbf{x}, \boldsymbol{\theta}_t, \boldsymbol{\theta}_r)   \\[-0.5em]
      \label{costraint:amp}
      \mathrm{s.t.} \quad & \beta_{t,n}^2 + \beta_{r,n}^2 = 1, \forall n \in \mathcal{N}, \\
      \label{costraint:phase}
      & \cos(\phi_{t,n} -  \phi_{r,n}) = 0, \forall n \in \mathcal{N},
  \end{align}
\end{subequations}
where $\mathcal{X}$ denotes the convex feasible set of optimization variable $\mathbf{x}$, vector $\boldsymbol{\theta}_{i} = [\beta_{i,1} e^{j \phi_{i,1}},$ $\dots, \beta_{i,N} e^{j \phi_{i,N}}]^T, \forall i \in \{t,r\},$ contains the transmission and reflection coefficients, and $F(\mathbf{x}, \boldsymbol{\theta}_t, \boldsymbol{\theta}_r)$ is the convex objective function or the convex approximation of the original objective function for the specific problem considered. Depending on the application scenario, $F(\cdot)$ can be some utility function such as transmit power, weighted sum-rate, harvested energy, and sensing accuracy, and $\mathbf{x}$ can represent the transmit waveforms or resource allocation variables. Note that in the above problem, the constraints $\beta_{t,n} \in [0,1]$ and $\beta_{r,n} \in [0,1]$ are omitted without changing the optimal objective value. This is because for any optimal solution $(\beta_{i,n}^\star, \phi_{i,n}^\star), \forall i \in \{t,r\}$, of problem \eqref{problem:original} with $\beta_{i,n}^\star \in [-1,0]$, it can be verified that $(-\beta_{i,n}^\star, \phi_{i,n}^\star + \pi)$ is also a feasible solution achieving the same objective value.

In problem \eqref{problem:original}, the non-convexity is caused by non-convex STAR-RIS constraints \eqref{costraint:amp} and \eqref{costraint:phase}. To tackle this challenge, we define the auxiliary variables $\tilde{\boldsymbol{\theta}}_{i} = [\tilde{\beta}_{i,1} e^{j \tilde{\phi}_{i,1}}, \dots, \tilde{\beta}_{i,N} e^{j \tilde{\phi}_{i,N}}]^T, \forall i \in \{t,r\},$ such that $\tilde{\boldsymbol{\theta}}_i = \boldsymbol{\theta}_i, \forall i \in \{t,r\}$. 
Then, problem \eqref{problem:original} can be rewritten as:
\begin{subequations} \label{problem:tranformed}
  \begin{align}
      \min_{\mathbf{x} \in \mathcal{X}, \boldsymbol{\theta}_t, \boldsymbol{\theta}_r, \tilde{\boldsymbol{\theta}}_t, \tilde{\boldsymbol{\theta}}_r} & F(\mathbf{x}, \boldsymbol{\theta}_t, \boldsymbol{\theta}_r)   \\[-0.5em]
      \label{constraint:equality}
      \mathrm{s.t.} \quad &\tilde{\boldsymbol{\theta}}_t = \boldsymbol{\theta}_t, \tilde{\boldsymbol{\theta}}_r = \boldsymbol{\theta}_r, \\
      \label{constraint:tranformed_3}
      & \tilde{\beta}_{t,n}^2 + \tilde{\beta}_{r,n}^2 = 1, \forall n \in \mathcal{N}, \\
      \label{constraint:tranformed_4}
      & \cos(\tilde{\phi}_{t,n} -  \tilde{\phi}_{r,n}) = 0, \forall n \in \mathcal{N}.
  \end{align}
\end{subequations}  
In the above problem, there are no constraints imposed on $\boldsymbol{\theta}_t$ and $\boldsymbol{\theta}_r$, except for equality constraint \eqref{constraint:equality}. To handle this equality constraint, we exploit the PDD framework \cite{shi2020penalty}, where the original problem is converted to the corresponding augmented Lagrangian (AL) problem by moving the equality constraints as a penalty term to the objective function. The AL problem corresponding to \eqref{problem:tranformed} is given by
\begin{subequations} \label{problem:penalty}
  \begin{align}
      \min_{ \scriptstyle \mathbf{x} \in \mathcal{X}, \boldsymbol{\theta}_t, \boldsymbol{\theta}_r, \atop \scriptstyle \tilde{\boldsymbol{\theta}}_t, \tilde{\boldsymbol{\theta}}_r} & F(\mathbf{x}, \boldsymbol{\theta}_t, \boldsymbol{\theta}_r) + \frac{1}{2 \rho} \sum_{i \in \{t,r\}} \|\tilde{\boldsymbol{\theta}}_i - \boldsymbol{\theta}_i + \rho \boldsymbol{\lambda}_i \|^2  \\[-0.5em]
      \mathrm{s.t.} \quad & \eqref{constraint:tranformed_3}, \eqref{constraint:tranformed_4},
  \end{align}
\end{subequations} 
where $\rho > 0$ denotes the penalty factor penalizing the violation of constraint \eqref{constraint:equality} and $\boldsymbol{\lambda}_i, \forall i \in \{t,r\},$ denotes the Lagrangian dual variables. As can be observed, when $\rho \rightarrow 0$, then the penalty term will be forced to zero, i.e., equality constraint \eqref{constraint:equality} is enforced. It has been proved that updating the primal and dual variables as well as the penalty factor in an alternating manner, a KKT optimal solution can be obtained via PDD under some mild conditions, such as the \emph{Robinson’s condition} or the \emph{Mangasarian-Fromovitz constraint qualification (MFCQ)} condition \cite{shi2020penalty}. Thus, in the following, we focus on solving problem \eqref{problem:penalty} by invoking the block successive upper-bound minimization (BSUM) or block coordinate descent (BCD) methods\footnote{For BSUM, a locally tight upper bound of the original objective function is optimized in each block. BSUM reduces to BCD when the upper bound is replaced by the original objective function itself \cite{shi2020penalty2}.}, where we divide the optimization variables into two blocks, namely $\{\mathbf{x}, \boldsymbol{\theta}_t, \boldsymbol{\theta}_r\}$ and $\{\tilde{\boldsymbol{\theta}}_t, \tilde{\boldsymbol{\theta}}_r\}$. 

\subsubsection{Subproblem with respect to $\{\mathbf{x}, \boldsymbol{\theta}_t, \boldsymbol{\theta}_r\}$}
Note that introducing the penalty term has no influence on the convexity of the objective function.
Thus, the subproblem with respect to $\{\mathbf{x}, \boldsymbol{\theta}_t, \boldsymbol{\theta}_r\}$ is convex and given by 
\begin{align} \label{problem:primal}
    \min_{\mathbf{x} \in \mathcal{X}, \boldsymbol{\theta}_t, \boldsymbol{\theta}_r} & F(\mathbf{x}, \boldsymbol{\theta}_t, \boldsymbol{\theta}_r) + \frac{1}{2 \rho} \sum_{i \in \{t,r\}} \|\tilde{\boldsymbol{\theta}}_i - \boldsymbol{\theta}_i + \rho \boldsymbol{\lambda}_i \|^2.
\end{align}
As a consequence, the optimal solution to the above problem can be efficiently obtained. 

\subsubsection{Subproblem with respect to $\{\tilde{\boldsymbol{\theta}}_t, \tilde{\boldsymbol{\theta}}_r\}$}
Since variables $\{\tilde{\boldsymbol{\theta}}_t, \tilde{\boldsymbol{\theta}}_r\}$ only appear in the constraints and in the penalty term of the objective function, the resulting problem is given by 
\begin{subequations} \label{problem:auxiliary}
  \begin{align}
      \min_{\tilde{\boldsymbol{\theta}}_t, \tilde{\boldsymbol{\theta}}_r} & \sum_{i \in \{t,r\}} \|\tilde{\boldsymbol{\theta}}_i + \boldsymbol{\vartheta}_i \|^2  \\
      \mathrm{s.t.} \quad &\tilde{\beta}_{t,n}^2 + \tilde{\beta}_{r,n}^2 = 1, \forall n \in \mathcal{N}, \\
      & \cos(\tilde{\phi}_{t,n} -  \tilde{\phi}_{r,n}) = 0, \forall n \in \mathcal{N},
  \end{align}
\end{subequations} 
where $\boldsymbol{\vartheta}_i = -\boldsymbol{\theta}_i + \rho \boldsymbol{\lambda}_i, \forall i \in \{t,r\}$. Although the constraints of the above optimization problem are non-convex, we show that a high-quality solution to this problem can be obtained with low complexity by alternately optimizing the amplitudes and phase shifts. Firstly, the objective function of problem \eqref{problem:auxiliary} can be reformulated as follows:
\begin{align}
  &\sum_{i \in \{t,r\}} \|\tilde{\boldsymbol{\theta}}_i + \boldsymbol{\vartheta}_i \|^2 = \sum_{i \in \{t,r\}} (\tilde{\boldsymbol{\theta}}_i^H \tilde{\boldsymbol{\theta}}_i +  \boldsymbol{\vartheta}_i^H \boldsymbol{\vartheta}_i + 2 \mathrm{Re}( \boldsymbol{\vartheta}_i^H  \tilde{\boldsymbol{\theta}}_i) ) \nonumber \\
  = &\sum_{i \in \{t,r\}} \sum_{n \in \mathcal{N}} \tilde{\beta}_{i,n}^2 + \sum_{i \in \{t,r\}} \boldsymbol{\vartheta}_i^H \boldsymbol{\vartheta}_i + \sum_{i \in \{t,r\}} 2 \mathrm{Re}( \boldsymbol{\vartheta}_i^H  \tilde{\boldsymbol{\theta}}_i) \nonumber \\
  \overset{(a)}{=} & N + \sum_{i \in \{t,r\}} \boldsymbol{\vartheta}_i^H \boldsymbol{\vartheta}_i + \sum_{i \in \{t,r\}} 2 \mathrm{Re}( \boldsymbol{\vartheta}_i^H  \tilde{\boldsymbol{\theta}}_i),
\end{align}
where $(a)$ is due to constraint $\tilde{\beta}_{t,n}^2 + \tilde{\beta}_{r,n}^2 = 1$. In the objective function, only the term $\sum_{i \in \{t,r\}} 2 \mathrm{Re}( \boldsymbol{\vartheta}_i^H  \tilde{\boldsymbol{\theta}}_i)$ involves the optimization variables, while the other terms are constants. Then, we decompose $\tilde{\boldsymbol{\theta}}_i$ to amplitude vector $\tilde{\boldsymbol{\beta}}_i = [\tilde{\beta}_{i,1},\dots,\tilde{\beta}_{i,N}]^T$ and phase-shift vector $\tilde{\boldsymbol{\psi}}_i = [e^{j \tilde{\phi}_{i,1}},\dots,e^{j \tilde{\phi}_{i,N}}]^T$, i.e.,
\begin{equation}
  \tilde{\boldsymbol{\theta}}_i = \mathrm{diag} (\tilde{\boldsymbol{\beta}}_i) \tilde{\boldsymbol{\psi}}_i = \mathrm{diag} (\tilde{\boldsymbol{\psi}}_i) \tilde{\boldsymbol{\beta}}_i, \forall i \in \{t,r\}.
\end{equation}
Consequently, problem \eqref{problem:auxiliary} can be rewritten as follows:
\begin{subequations} \label{problem:closed-form}
  \begin{align}
      &\min_{\tilde{\boldsymbol{\beta}}_t, \tilde{\boldsymbol{\psi}}_t, \tilde{\boldsymbol{\beta}}_r, \tilde{\boldsymbol{\psi}}_r} \sum_{i \in \{t,r\}} \mathrm{Re}(\boldsymbol{\vartheta}_i \mathrm{diag} (\tilde{\boldsymbol{\beta}}_i) \tilde{\boldsymbol{\psi}}_i )  \\
      \mathrm{s.t.} \quad &\tilde{\beta}_{t,n}^2 + \tilde{\beta}_{r,n}^2 = 1, 0 \le \tilde{\beta}_{t,n}, \tilde{\beta}_{r,n} \le 1, \forall n \in \mathcal{N}, \\
      &\cos(\tilde{\phi}_{t,n} -  \tilde{\phi}_{r,n}) = 0, \forall n \in \mathcal{N}.
  \end{align}
\end{subequations} 
Here, we introduce the constraint $0 \le \tilde{\beta}_{t,n}, \tilde{\beta}_{r,n} \le 1$ back to ensure that all the phase shifts are collected in $\tilde{\phi}_{t,n}$ and $\tilde{\phi}_{t,n}$ during the following optimization process, which does not affect the optimal objective value based on the previous analysis. To solve the above optimization problem, we introduce the following two propositions.
\begin{proposition} \label{proposition_1}
  \emph{
  (Closed-form optimal solution for phase shifts for given amplitudes) 
  Define $\tilde{\boldsymbol{\vartheta}}_i = \mathrm{diag}(\tilde{\boldsymbol{\beta}}_i^H)\boldsymbol{\vartheta}_i = [\tilde{\vartheta}_{i,1}, \dots, \tilde{\vartheta}_{i,N}]^T$, $\varphi_n^+ = \tilde{\vartheta}_{t,n}^* + j \tilde{\vartheta}_{r,n}^*$, and $\varphi_n^- = \tilde{\vartheta}_{t,n}^* - j \tilde{\vartheta}_{r,n}^*$. Then, for any given $\tilde{\boldsymbol{\beta}}_t$ and $\tilde{\boldsymbol{\beta}}_r$, the optimal solutions for the elements of $\tilde{\boldsymbol{\psi}}_t$ and $\tilde{\boldsymbol{\psi}}_r$ are given by
  \begin{equation} \label{eqn:optimal_psi}
    (\tilde{\psi}_{t,n}^\star, \tilde{\psi}_{r,n}^\star) = \!\!\!\!\operatorname*{argmin}_{(\tilde{\psi}_{t,n}, \tilde{\psi}_{r,n}) \in \chi_{\psi}^n} \!\!\!\! \mathrm{Re}(\tilde{\vartheta}_{t,n}^* \tilde{\psi}_{t,n} ) + \mathrm{Re}(\tilde{\vartheta}_{r,n}^* \tilde{\psi}_{r,n} ),
  \end{equation} 
  where $\chi_{\psi}^n$ denotes a set of a pair of closed-form solutions:
  \begin{equation}
    \chi_{\psi}^n = \left\{ (e^{j (\pi - \angle \varphi_n^+ )}, e^{j (\frac{3}{2}\pi - \angle \varphi_n^+ )}), (e^{j (\pi - \angle \varphi_n^- )}, e^{j (\frac{1}{2}\pi - \angle \varphi_n^- )}) \right\}.
  \end{equation} 
  }
\end{proposition}

\begin{proof}
  Please refer to Appendix A.
\end{proof}

\begin{proposition} \label{proposition_2}
  \emph{
  (Closed-form optimal solution for amplitudes for given phase shifts) 
  Define $\breve{\boldsymbol{\vartheta}}_i = \mathrm{diag}(\tilde{\boldsymbol{\psi}}_i^H)\boldsymbol{\vartheta}_i = [\breve{\vartheta}_{i,1}, \dots, \breve{\vartheta}_{i,N}]^T$,  
  $a_n = |\breve{\vartheta}_t^*| \cos(\angle \breve{\vartheta}_t^* )$, $b_n = |\breve{\vartheta}_r^*| \sin(\angle \breve{\vartheta}_r^* )$, $\xi _n = \mathrm{sgn}(b_n) \arccos(\frac{a_n}{\sqrt{a_n^2 + b_n^2}})$. Then, for any given $\tilde{\boldsymbol{\psi}}_t$ and $\tilde{\boldsymbol{\psi}}_r$, the optimal solutions for the elements of $\tilde{\boldsymbol{\beta}}_t$ and $\tilde{\boldsymbol{\beta}}_r$ are given by 
  \begin{equation} \label{eqn:optimal_beta}
    \tilde{\beta}_{t,n}^\star = \sin \omega_n, \quad \tilde{\beta}_{r,n}^\star = \cos \omega_n,
  \end{equation}
  where 
  \begin{equation} \label{eqn:optimal_omega_0}
    \omega_n = \begin{cases}
        -\frac{1}{2}\pi - \xi_n, &\text{if } \xi_n \in [-\pi, -\frac{1}{2}\pi),\\[-0.5em]
        0, & \text{if } \xi_n \in [-\frac{1}{2}\pi, \frac{1}{4}\pi), \\[-0.5em]
        \frac{1}{2}\pi, &\text{otherwise}.
    \end{cases}
  \end{equation}  
  }
\end{proposition}

\begin{proof}
  Please refer to Appendix B.
\end{proof}

\begin{algorithm}[tb]
  \caption{BSUM/BCD algorithm for solving problem \eqref{problem:penalty}.}
  \label{alg:closed_form_AO}
  \begin{algorithmic}[1]
      \State{Initialize the optimization variables}
      \Repeat
      \State{update $\{\mathbf{x}, \boldsymbol{\theta}_t, \boldsymbol{\theta}_r\}$ by solving problem \eqref{problem:primal}}
      \State{update $\{\tilde{\boldsymbol{\psi}}_t, \tilde{\boldsymbol{\psi}}_r\}$ by \eqref{eqn:optimal_psi}}
      \State{update $\{\tilde{\boldsymbol{\beta}}_t, \tilde{\boldsymbol{\beta}}_r\}$ by \eqref{eqn:optimal_beta}}
      \Until{convergence}
  \end{algorithmic}
\end{algorithm}

According to \textbf{Propositions \ref{proposition_1}} and \textbf{\ref{proposition_2}}, we can further divide block $\{\tilde{\boldsymbol{\theta}}_t, \tilde{\boldsymbol{\theta}}_r\}$ into two sub-blocks, namely $\{\tilde{\boldsymbol{\psi}}_t, \tilde{\boldsymbol{\psi}}_r\}$ and $\{\tilde{\boldsymbol{\beta}}_t, \tilde{\boldsymbol{\beta}}_r\}$. Then, the overall BSUM/BCD algorithm to solve problem \eqref{problem:penalty} is summarized in \textbf{Algorithm \ref{alg:closed_form_AO}}. Since the optimal solution is obtained in each step, the convergence of \textbf{Algorithm \ref{alg:closed_form_AO}} is guaranteed \cite{razaviyayn2013unified}. The complexities of updating $\{\tilde{\boldsymbol{\psi}}_t, \tilde{\boldsymbol{\psi}}_r\}$ and $\{\tilde{\boldsymbol{\beta}}_t, \tilde{\boldsymbol{\beta}}_r\}$ are $\mathcal{O}(4N)$ and $\mathcal{O}(2N)$, respectively, where $\mathcal{O}(\cdot)$ is the big-O notation. Moreover, the complexity of updating $\{\mathbf{x}, \boldsymbol{\theta}_t, \boldsymbol{\theta}_r\}$ is determined by the exact form of problem \eqref{problem:primal} and the methods used to solve it.

\section{Case Study and Numerical Results}
To verify the effectiveness of the proposed general optimization framework, in this section, we use a case study, where we maximize the throughput in a narrowband STAR-RIS-aided communication system.

\subsection{System Model and Problem Formulation}
Consider an $M$-antenna base station (BS), an $N$-element STAR-RIS, and $K$ single-antenna users, whose indices are collected in $\mathcal{K}$. Without loss of generality, we assume that the users in subset $\mathcal{K}_t = \{1,\dots,K_0\}$ are located on the transmission side, and the users in subset $\mathcal{K}_r = \{K_0+1,\dots,K\}$ are located on the reflection side. The direct BS-user channels are assumed to be blocked. Thus, the received signal at user $k, \forall k \in \mathcal{K}_i, \forall i \in \{t,r\},$ is given by  
\begin{equation}
  y_k =  \mathbf{h}_k^H \mathbf{\Theta}_i \mathbf{G} \sum_{\ell \in \mathcal{K}} \mathbf{w}_\ell s_\ell + n_k,
\end{equation}
where $\mathbf{h}_k \in \mathbb{C}^{N \times 1}$ denotes the STAR-RIS-user-$k$ channel, $\mathbf{G} \in \mathbb{C}^{M \times N}$ denotes the BS-STAR-RIS channel, $\mathbf{\Theta}_i = \mathrm{diag}(\boldsymbol{\theta}_i)$ denotes the transmission/reflection coefficients of the STAR-RIS, $\mathbf{w}_\ell \in \mathbb{C}^{M \times 1}$ denotes the active beamforming vector at the BS for delivering information symbol $s_\ell \in \mathbb{C}$ to user $\ell$, and $n_k \sim \mathcal{CN}(0, \sigma_k^2)$ denotes complex Gaussian noise with power $\sigma_k^2$. Then, the signal-to-interference-plus-noise ratio (SINR) for decoding $s_k$ at user $k, \forall k \in \mathcal{K}_i, \forall i \in \{t,r\},$ is given by 
\begin{align}
  \gamma_k = &\frac{|\mathbf{h}_k^H \mathbf{\Theta}_i \mathbf{G} \mathbf{w}_k|^2}{\sum_{\ell \in \mathcal{K} \setminus k} |\mathbf{h}_k^H \mathbf{\Theta}_i \mathbf{G} \mathbf{w}_\ell|^2 + \sigma_k^2}.
\end{align}
The corresponding achievable rate is $R_k = \log_2(1 + \gamma_k)$. 

We aim to maximize the throughput of the considered STAR-RIS-aided system subject to a transmit power constraint and coupled STAR-RIS phase-shift and amplitude constraints. The corresponding optimization problem can be formulated as follows:
\begin{subequations} \label{problem:throughput}
  \begin{align}
      \max_{\mathbf{W}, \boldsymbol{\theta}_t, \boldsymbol{\theta}_r} \quad & \sum_{k \in \mathcal{K}} R_k \\
      \label{constraint:throught_1}
      \mathrm{s.t.} \quad & \mathrm{tr}(\mathbf{W} \mathbf{W}^H) \le P_t, \\
      \label{constraint:throught_2}
      & \beta_{t,n}^2 + \beta_{r,n}^2 = 1, \forall n \in \mathcal{N}, \\
      \label{constraint:throught_3}
      & \cos(\phi_{t,n} -  \phi_{r,n}) = 0, \forall n \in \mathcal{N},
  \end{align}
\end{subequations}
where $\mathbf{W} = [\mathbf{w}_1, \dots, \mathbf{w}_K]$ and $P_t$ denotes the BS transmit power. We note that existing methods for solving throughput-maximization problems cannot be directly applied to problem \eqref{problem:throughput} due to the coupled STAR-RIS phase-shift and amplitude constraints. In the following section, we show that the proposed general optimization framework can be used to effectively solve problem \eqref{problem:throughput}.

\subsection{Solution to Problem \eqref{problem:throughput} using the Proposed Framework}
Note that the objective function of problem \eqref{problem:throughput} is non-convex with respect to $\{\mathbf{W}, \boldsymbol{\theta}_t, \boldsymbol{\theta}_r\}$. In order to employ the proposed optimization framework, we first transform the throughput maximization problem into an equivalent weighted mean square error (MSE) minimization problem applying the well-known weighted minimum mean square error (WMMSE) method \cite{christensen2008weighted} as follows:
\begin{subequations} \label{problem:WMMSE_0}
  \begin{align}
      \max_{\boldsymbol{\varpi }, \boldsymbol{\upsilon}, \mathbf{W}, \boldsymbol{\theta}_t, \boldsymbol{\theta}_r}  & \sum_{k \in \mathcal{K}} \varpi_k e_k \\
      \mathrm{s.t.} \quad & \eqref{constraint:throught_1} - \eqref{constraint:throught_3}.
  \end{align}
\end{subequations}
Here, $\boldsymbol{\varpi } = [\varpi _1,\dots,\varpi _K]^T$ denotes the vector of weights, and $e_k, \forall k \in \mathcal{K}_i, \forall i \in \{t,r\},$ denotes the MSE as follows: 
\begin{align}
  e_k = &|\upsilon_k|^2 \big( \sum_{\ell \in \mathcal{K}} |\boldsymbol{\theta}_i^T \mathrm{diag}(\mathbf{h}_k^H) \mathbf{G} \mathbf{w}_\ell|^2 + \sigma_k^2\big) \nonumber \\[-0.5em]
  &- 2 \mathrm{Re} \{ \upsilon_k^* \boldsymbol{\theta}_i^T \mathrm{diag}(\mathbf{h}_k^H) \mathbf{G} \mathbf{w}_k \} + 1,
\end{align} 
where $\boldsymbol{\upsilon} = [\upsilon_1,\dots,\upsilon_K]^T$ are auxiliary variables. According to \cite{christensen2008weighted}, it can be proved that if $\{\boldsymbol{\varpi }, \boldsymbol{\upsilon}, \mathbf{W}, \boldsymbol{\theta}_t, \boldsymbol{\theta}_r\}$ is a KKT optimal solution to problem \eqref{problem:WMMSE_0}, $\{\mathbf{W}, \boldsymbol{\theta}_t, \boldsymbol{\theta}_r\}$ is also a KKT optimal solution to problem \eqref{problem:throughput}. In problem \eqref{problem:WMMSE_0}, the objective function is block-wise convex with respect to $\{\boldsymbol{\varpi }, \boldsymbol{\upsilon}\}$, $\mathbf{W}$, and $\{\boldsymbol{\theta}_t, \boldsymbol{\theta}_r\}$. Moreover, the feasible sets of $\{\boldsymbol{\varpi }, \boldsymbol{\upsilon}\}$ and $\mathbf{W}$ are also convex. Thus, we have transformed the throughput maximization problem into the canonical form of problem \eqref{problem:original}, where the corresponding optimization variables, feasible set, and objective function are given by $\mathbf{x} = \{\boldsymbol{\varpi }, \boldsymbol{\upsilon}, \mathbf{W}\}$, $\mathcal{X} = \left\{(\boldsymbol{\varpi }, \boldsymbol{\upsilon}, \mathbf{W}) | \mathrm{tr}(\mathbf{W}\mathbf{W}^H) \le P_t \right\}$, and $F(\mathbf{x}, \boldsymbol{\theta}_t, \boldsymbol{\theta}_r) = \sum_{k \in \mathcal{K}} \varpi_k e_k$, respectively. Therefore, we can employ the proposed framework to solve problem \eqref{problem:throughput}. 

By defining $\tilde{\boldsymbol{\theta}}_t = \boldsymbol{\theta}_t$ and $\tilde{\boldsymbol{\theta}}_r = \boldsymbol{\theta}_r$, problem \eqref{problem:WMMSE_0} can be equivalently reformulated as follows: 
\begin{subequations} \label{problem:WMMSE}
  \begin{align}
      \min_{\scriptstyle \boldsymbol{\varpi }, \boldsymbol{\upsilon}, \mathbf{W}, \boldsymbol{\theta}_t, \boldsymbol{\theta}_r \atop  \scriptstyle \tilde{\boldsymbol{\theta}}_t, \tilde{\boldsymbol{\theta}}_r } & \sum_{k \in \mathcal{K}} \varpi_k e_k \\[-0.5em]
      \label{constraint:WMMSE_1}
      \mathrm{s.t.} \quad &\tilde{\boldsymbol{\theta}}_t = \boldsymbol{\theta}_t, \tilde{\boldsymbol{\theta}}_r = \boldsymbol{\theta}_r \\
      \label{constraint:WMMSE_2}
      &\mathrm{tr}(\mathbf{W} \mathbf{W}^H) \le P_t, \\
      \label{constraint:WMMSE_3}
      & \tilde{\beta}_{t,n}^2 + \tilde{\beta}_{r,n}^2 = 1, \forall n \in \mathcal{N}, \\
      \label{constraint:WMMSE_4}
      & \cos(\tilde{\phi}_{t,n} -  \tilde{\phi}_{r,n}) = 0, \forall n \in \mathcal{N}.
  \end{align}
\end{subequations}
By moving equality constraint \eqref{constraint:WMMSE_1} via a penalty term to the objective function, the following problem is obtained:
\begin{subequations} \label{problem:final}
  \begin{align}
      \min_{\scriptstyle \boldsymbol{\varpi }, \boldsymbol{\upsilon}, \mathbf{W}, \boldsymbol{\theta}_t, \boldsymbol{\theta}_r \atop \scriptstyle \tilde{\boldsymbol{\theta}}_t, \tilde{\boldsymbol{\theta}}_r} & \sum_{k \in \mathcal{K}} \varpi _k e_k + \frac{1}{2 \rho} \sum_{i \in \{t,r\}} \|\tilde{\boldsymbol{\theta}}_i - \boldsymbol{\theta}_i + \rho \boldsymbol{\lambda}_i \|^2 \\
      \mathrm{s.t.} \quad & \eqref{constraint:WMMSE_2} - \eqref{constraint:WMMSE_4}.
  \end{align}
\end{subequations}
Then, the above problem can be solved via BCD by alternately optimizing the blocks $\{\boldsymbol{\varpi }, \boldsymbol{\upsilon}\}$, $\mathbf{W}$, $\{\boldsymbol{\theta}_t, \boldsymbol{\theta}_r\}$, $\{\tilde{\boldsymbol{\psi}}_t, \tilde{\boldsymbol{\psi}}_r\}$, and $\{\tilde{\boldsymbol{\beta}}_t, \tilde{\boldsymbol{\beta}}_r\}$.

\subsubsection{Subproblem with respect to $\{\boldsymbol{\varpi }, \boldsymbol{\upsilon}\}$}
The optimal $\varpi _k$ and $\upsilon_k, \forall k \in \mathcal{K}_i, \forall i \in \{t,r\},$ of this subproblem are given by \cite{christensen2008weighted}
\begin{align} 
  \label{eqn:optimal_omega}
  &\varpi _k = 1 + \gamma_k, \\
  \label{eqn:optimal_upsilon}
  &\upsilon_k = \frac{ \boldsymbol{\theta}_i^T \mathrm{diag}(\mathbf{h}_k^H) \mathbf{G} \mathbf{w}_k }{\sum_{\ell \in \mathcal{K}} |\boldsymbol{\theta}_i^T \mathrm{diag}(\mathbf{h}_k^H) \mathbf{G} \mathbf{w}_\ell|^2 + \sigma_k^2}.
\end{align}

\subsubsection{Subproblems with respect to $\mathbf{W}$ and $\{\boldsymbol{\theta}_t, \boldsymbol{\theta}_r\}$}
Note that the objective function of \eqref{problem:WMMSE} is convex with respect to $\mathbf{W}$ and $\{\boldsymbol{\theta}_t, \boldsymbol{\theta}_r\}$, respectively. Thus, the corresponding optimal solution can be efficiently obtained using existing optimization toolboxes, such as CVX \cite{cvx}. 

\subsubsection{Subproblems with respect to $\{\tilde{\boldsymbol{\psi}}_t, \tilde{\boldsymbol{\psi}}_r\}$ and $\{\tilde{\boldsymbol{\beta}}_t, \tilde{\boldsymbol{\beta}}_r\}$}
The optimal $\{\tilde{\boldsymbol{\psi}}_t, \tilde{\boldsymbol{\psi}}_r\}$ and $\{\tilde{\boldsymbol{\beta}}_t, \tilde{\boldsymbol{\beta}}_r\}$ can directly obtained based on \textbf{Propositions \ref{proposition_1}} and \textbf{\ref{proposition_2}}.

Finally, the dual variables $\{\boldsymbol{\lambda}_t, \boldsymbol{\lambda}_r\}$ and penalty factor $\rho$ can be updated following the PDD framework. As a consequence, the overall algorithm for solving problem \eqref{problem:WMMSE} is summarized in \textbf{Algorithm \ref{alg:throughput}}, where $\delta = \max \{\| \tilde{\boldsymbol{\theta}}_t - \boldsymbol{\theta}_t \|_{\infty}, \| \tilde{\boldsymbol{\theta}}_r - \boldsymbol{\theta}_r \|_{\infty} \}$ denotes the constraint violation function. If the MFCQ condition is satisfied, the PDD framework can obtain a KKT optimal solution to problem \eqref{problem:WMMSE}, c.f. \cite{shi2020penalty}. Thus, we show that the MFCQ condition indeed holds for problem \eqref{problem:WMMSE}.

\begin{algorithm}[tb]
  \caption{PDD-based algorithm for solving problem \eqref{problem:WMMSE}.}
  \label{alg:throughput}
  \begin{algorithmic}[1]
      \State{Initialize the optimization variables, and $0<c<1$}
      \Repeat
      \Repeat
      \State{update $\{\boldsymbol{\varpi }, \boldsymbol{\upsilon}\}$ by \eqref{eqn:optimal_omega} and \eqref{eqn:optimal_upsilon}}
      \State{update $\mathbf{W}$ by solving \eqref{problem:final} for $\mathbf{W}$}
      \State{update $\{\boldsymbol{\theta}_t, \boldsymbol{\theta}_r\}$ by solving \eqref{problem:final} for $\{\boldsymbol{\theta}_t, \boldsymbol{\theta}_r\}$}
      \State{update $\{\tilde{\boldsymbol{\psi}}_t, \tilde{\boldsymbol{\psi}}_r\}$ by \eqref{eqn:optimal_psi}}
      \State{update $\{\tilde{\boldsymbol{\beta}}_t, \tilde{\boldsymbol{\beta}}_r\}$ by \eqref{eqn:optimal_beta}}
      \Until{convergence}
      \If {$\delta \le \eta$} set $\boldsymbol{\lambda}_i = \boldsymbol{\lambda}_i + \frac{1}{\rho} ( \tilde{\boldsymbol{\theta}}_i - \boldsymbol{\theta}_i ), \forall i \in \{t,r\} $
      \Else {} set $\rho = c \rho$
      \EndIf
      \State{set $\eta = 0.9 \delta$}
      \Until{$\delta$ falls below a predefined threshold}
  \end{algorithmic}
\end{algorithm}

\begin{proposition} \label{proposition_3}
  \emph{MFCQ holds for problem \eqref{problem:WMMSE} at any feasible point $\{\mathbf{W}, \boldsymbol{\theta}_t, \boldsymbol{\theta}_r\}$ with $\mathbf{W} \neq \mathbf{0}$.}
\end{proposition}

\begin{proof}
  Please refer to Appendix C.
\end{proof}
According to \textbf{Proposition \ref{proposition_3}}, we can conclude that a KKT optimal solution to problem \eqref{problem:WMMSE} can be obtained with \textbf{Algorithm \ref{alg:throughput}}, which is also a KKT optimal solution to the original problem \eqref{problem:throughput}.
\begin{remark}
  \emph{
    This case study reveals that once the original problem is transformed into a form for which 1) the objective function is convex or block-wise convex, and 2) the feasible set is convex or block-wise convex except for constraints \eqref{costraint:amp} and \eqref{costraint:phase}, the proposed optimization framework can be directly used. Typically, such a transformation can be achieved by existing methods such as WMMSE, majorization-minimization, and successive convex optimization (SCA). In most cases, if the transformed problem satisfies the mild MFCQ condition, the KKT optimal solution can be obtained when PDD is employed to update the dual variables and the penalty factor in the proposed framework \cite{shi2020penalty}. Otherwise, at least the convergence of the proposed framework can be guaranteed.
  }
\end{remark}

\subsection{Numerical Results}
In this section, simulation results are provided to verify the effectiveness of the proposed optimization framework. Here, we assume that the BS with $M=8$ antennas is $50$ meters, under an angle of $20^\circ$, away from the STAR-RIS. The users are located on half-circles centered at the STAR-RIS with a radius of $3$ meters. We also assume that half of the users are located on the transmission side and the rest are located on the reflection side. The channels are modeled as Rician fading channels with a Rician factor of $3$ dB and a path loss exponent of $2.2$. The path loss at the reference distance of $1$ meter is set to $30$ dB. The transmit power of the BS and the noise power at the users are set to $20$ dBm and $-110$ dBm, respectively. 

\begin{figure}[t!]
  \centering
  \begin{subfigure}[t]{0.4\textwidth}
    \centering
    \includegraphics[width=1\textwidth]{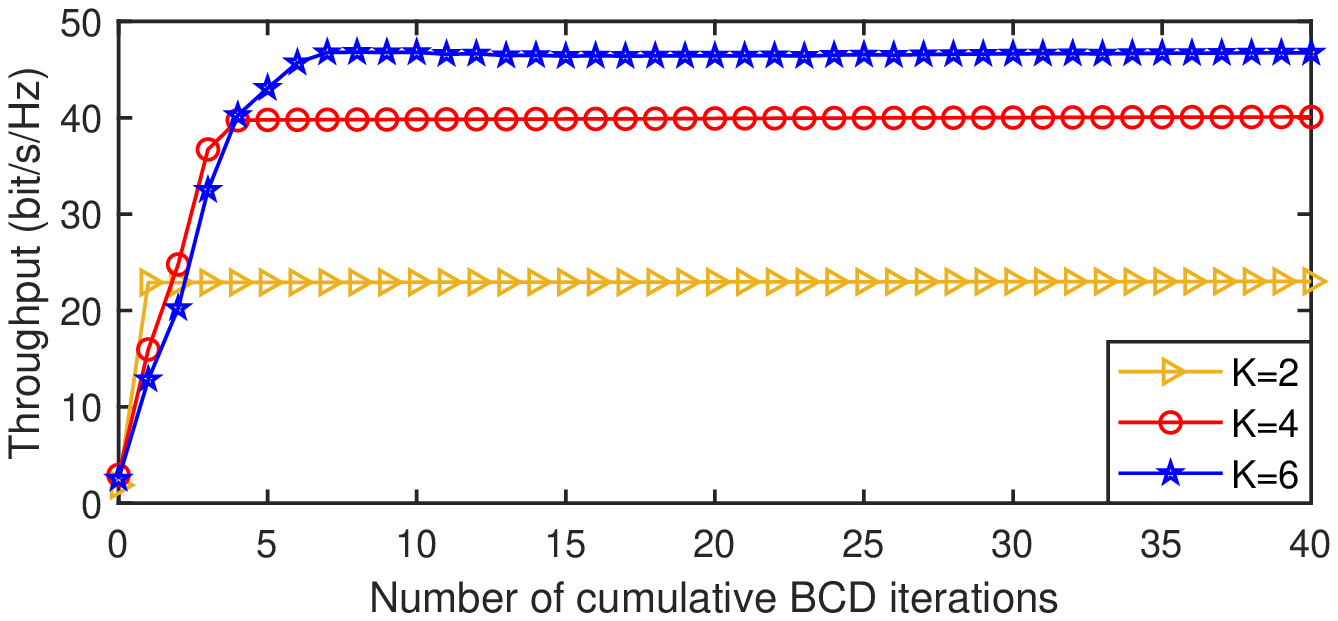}
    \subcaption{Throughput for different values of $K$. \label{fig:convergence_throughput}}
  \end{subfigure}
  \begin{subfigure}[t]{0.4\textwidth}
    \centering
    \includegraphics[width=1\textwidth]{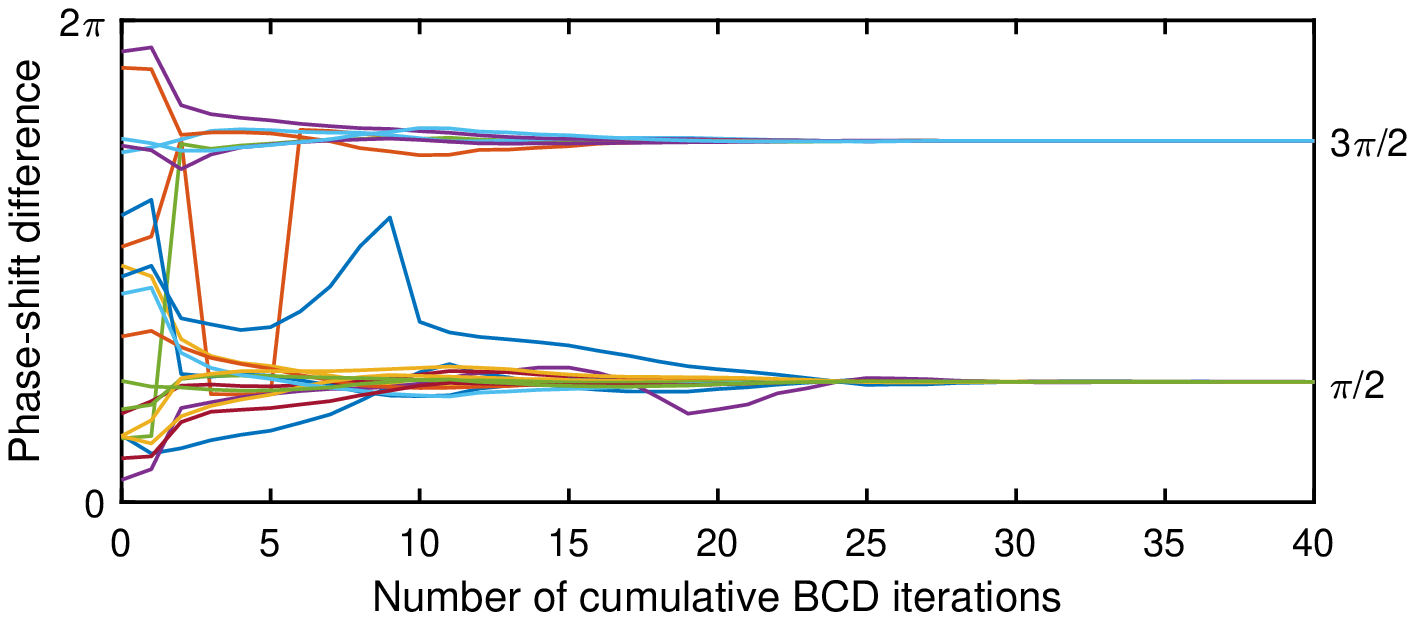}
    \subcaption{Phase-shift difference for $K=6$. \label{fig:convergence_phase_shift}}
  \end{subfigure}
  \caption{The convergence behavior of Algorithm \ref{alg:throughput}.}
  \label{fig:convergence}
\end{figure}

\begin{figure}[t!] 
  \centering
  \includegraphics[width=0.4\textwidth]{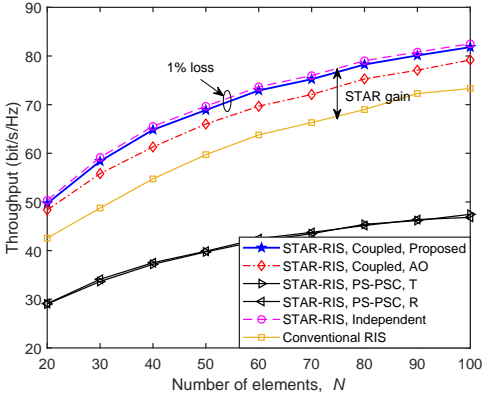}
  \caption{Throughput versus number of elements $N$.}
  \label{fig:rate_vs_elements}
\end{figure}

In Fig. \ref{fig:convergence}, we first study the convergence of the proposed \textbf{Algorithm \ref{alg:throughput}} for $N=20$. As can be seen from Fig. \ref{fig:convergence_throughput}, the throughput rapidly converges to a stationary value for all considered values of $K$. Moreover, Fig. \ref{fig:convergence_phase_shift} shows the absolute phase-shift difference $|\phi_{t,n} - \phi_{r,n}|$ for all $20$ elements. As can be observed, the phase-shift differences converge to $\frac{\pi}{2}$ or $\frac{3\pi}{2}$, i.e., $\cos(\phi_{t,n} - \phi_{r,n}) = 0$.

Next, we consider the following benchmark schemes for performance comparison. 
1) \textbf{STAR-RIS, Coupled, AO} \cite{niu2022weighted}: In this scheme, AO is exploited, where the coefficients of one side are optimized subject to the coupled phase-shift constraints by fixing the coefficients of the other side.
2) \textbf{STAR-RIS, PS-PSC, T} \cite{xu2022star}: This refers to the heuristic primary-secondary phase-shift configuration (PS-PSC) scheme, where the transmission side is primary for STAR-RIS. The STAR-RIS coefficients are obtained by fixing the optimal transmission coefficients and adjusting reflection coefficients such that the coupled phase-shift constraints are satisfied. 
3) \textbf{STAR-RIS, PS-PSC, R} \cite{xu2022star}: This scheme is similar to the previous scheme, but the reflection side is primary. 
4) \textbf{STAR-RIS, Independent} \cite{mu2021simultaneously}: In this scheme, the phase shift coefficients of the STAR-RIS can be independently adjusted. 
5) \textbf{Conventional RIS} \cite{huang2019reconfigurable}: In this scheme, there are two $\frac{N}{2}$-element reflect-only and transmit-only RISs deployed adjacent to each other.

In Fig. \ref{fig:rate_vs_elements}, we show the throughput versus the number of elements $N$ for different STAR-RIS optimization schemes when $K=6$. The results are obtained by averaging over $100$ random channel realizations. As can be observed, the proposed PDD-based optimization framework significantly outperforms the AO-based algorithm and the heuristic PS-PSC schemes for coupled phase-shift STAR-RISs. Moreover, Fig. \ref{fig:rate_vs_elements} also reveals that the coupled phase-shift model achieves almost the same performance as the independent one and achieves a significant performance gain over the conventional RIS.

\section{Conclusions}
We proposed a general optimization framework for STAR-RISs with coupled phase shifts, which gives the provable optimal solution under some mild conditions. Then, as a case study, we investigated throughput maximization based on the proposed optimization framework, where the KKT optimal solution was obtained. Our numerical results confirmed the effectiveness of the proposed optimization method. The proposed framework can be extended to STAR-RIS design in various network architectures.

\section*{Appendix~A: Proof of Proposition \ref{proposition_1}} \label{appendix:optimal_phase}
For any given $\tilde{\boldsymbol{\beta}}_t$ and $\tilde{\boldsymbol{\beta}}_r$, problem \eqref{problem:closed-form} can be decomposed into a series of independent optimization problems for each pair of $(\tilde{\psi}_{t,n}, \tilde{\psi}_{r,n})$, which leads to
\begin{subequations} \label{problem:phase-shift}
    \begin{align}
        \min_{\tilde{\psi}_{t,n}, \tilde{\psi}_{r,n}} \quad &
        \mathrm{Re}( \tilde{\vartheta}_{t,n}^* \tilde{\psi}_{t,n} ) + \mathrm{Re}( \tilde{\vartheta}_{r,n}^* \tilde{\psi}_{r,n} )\\[-0.5em]
        \label{constraint:appdeiex_2}
        \mathrm{s.t.} \quad 
        & \cos(\tilde{\phi}_{t,n} -  \tilde{\phi}_{r,n}) = 0, \\
        & |\tilde{\psi}_{t,n}| = 1, |\tilde{\psi}_{r,n}| = 1.
    \end{align}
\end{subequations}
Note that the coupled phase-shift constraint \eqref{constraint:appdeiex_2} can be rewritten as $| \tilde{\phi}_{t,n} -  \tilde{\phi}_{r,n} | = \frac{\pi}{2} \text{ or } \frac{3\pi}{2}$, which is equivalent to 
\begin{equation}
  \tilde{\psi}_{r,n} = j \tilde{\psi}_{t,n} \text{ or } \tilde{\psi}_{r,n} = -j \tilde{\psi}_{t,n}. 
\end{equation}
Substituting the above constraint into the objective function, problem \eqref{problem:phase-shift} can be further simplified as $\min_{|\tilde{\psi}_{t,n}| = 1} \mathrm{Re}\left( (\tilde{\vartheta}_{t,n}^* \pm j \tilde{\vartheta}_{r,n}^*) \tilde{\psi}_{t,n} \right)$, the optimal solution to which can be readily obtained as follows:
\begin{equation}
    \tilde{\psi}_{t,n} = e^{j \left(\pi - \angle (\tilde{\vartheta}_{t,n}^* \pm j \tilde{\vartheta}_{r,n}^*) \right)}.
\end{equation}
Comparing the objective values for the above two solutions, the optimal solution to problem \eqref{problem:phase-shift} can be obtained, which completes the proof.

\section*{Appendix~B: Proof of Proposition \ref{proposition_2}} \label{appendix:optimal_amplitude}
For any given $\tilde{\boldsymbol{\psi}}_t$ and $\tilde{\boldsymbol{\psi}}_r$, problem \eqref{problem:closed-form} can be decomposed into a series of independent optimization problems for each pair of $(\tilde{\beta}_{t,n}, \tilde{\beta}_{r,n})$, which leads to
\begin{subequations} \label{problem:separate_phase}
    \begin{align}
        \min_{\tilde{\beta}_{t,n}, \tilde{\beta}_{r,n}} \quad &
        \mathrm{Re}( \breve{\vartheta}_{t,n}^* \tilde{\beta}_{t,n}  ) + \mathrm{Re}( \breve{\vartheta}_{r,n}^* \tilde{\beta}_{r,n}  ) \\[-0.5em]
        \label{constraint:amplitude}
        \mathrm{s.t.} \quad & \tilde{\beta}_{t,n}^2 + \tilde{\beta}_{r,n}^2 = 1, 0 \le \tilde{\beta}_{t,n}, \tilde{\beta}_{r,n} \le 1.
    \end{align}
\end{subequations} 
Since $\tilde{\beta}_{i,n}$ is real-valued, we simplify the objective function as follows: $a_n \tilde{\beta}_{t,n} + b_n \tilde{\beta}_{r,n}$, where $a_n = |\breve{\vartheta}_{t,n}^*| \cos (\angle \breve{\vartheta}_{t,n}^*)$ and $b_n = |\breve{\vartheta}_{r,n}^*| \cos (\angle \breve{\vartheta}_{r,n}^*)$. To solve the simplified problem, polar coordinates are used, where we set $\tilde{\beta}_{t,n} = \sin \omega_n$ and $\tilde{\beta}_{r,n} = \cos \omega_n$ with $\omega_n \in [0, \frac{1}{2}\pi]$. Hence, constraint \eqref{constraint:amplitude} is automatically satisfied. Based on this transformation, the objective function can be rewritten as 
\begin{align}
    & a_n \sin \omega_n + b_n \cos \omega_n \nonumber \\
    = &\sqrt{a_n^2 + b_n^2} \left( \cos \xi_n \sin \omega_n + \sin \xi_n \cos \omega_n \right) \nonumber \\
    = &\sqrt{a_n^2 + b_n^2} \sin (\omega_n + \xi_n),
\end{align}
where $\cos \xi_n = \frac{a_n}{\sqrt{a_n^2 + b_n^2}}$ and $\sin \xi_n = \frac{b_n}{\sqrt{a_n^2 + b_n^2}}$. Then, the optimal solution in \eqref{eqn:optimal_omega_0} for minimizing $\sin (\omega_n + \xi_n)$ with respect to $\omega_n \in [0, \frac{1}{2}\pi]$ can be readily obtained, which completes the proof.

\section*{Appendix~C: Proof of Proposition \ref{proposition_3}} 
In this appendix, we show that the MFCQ condition \cite[Appendix C]{shi2020penalty2} is satisfied by problem \eqref{problem:WMMSE}. Note that for problem \eqref{problem:WMMSE}, optimizing $\tilde{\boldsymbol{\theta}}_i$ is equivalent to optimizing its amplitudes $\tilde{\boldsymbol{\beta}}_i = [\tilde{\beta}_{i,1}, \dots, \tilde{\beta}_{i,N}]$ and phase shifts $\tilde{\boldsymbol{\phi}}_i = [\tilde{\phi}_1, \dots, \tilde{\phi}_N]$. Thus, all equality constants in problem \eqref{problem:WMMSE} can be written as
\begin{subequations}
  \begin{align}
    &\boldsymbol{\mu}_1 = \tilde{\boldsymbol{\beta}}_t \circ \tilde{\boldsymbol{\beta}}_t + \tilde{\boldsymbol{\beta}}_r \circ \tilde{\boldsymbol{\beta}}_r - \mathbf{1} = \mathbf{0}, \\[-0.5em]
    &\boldsymbol{\mu}_2 = \cos(\tilde{\boldsymbol{\phi}}_t - \tilde{\boldsymbol{\phi}}_r) = \mathbf{0}, \\[-0.5em]
    &\boldsymbol{\mu}_3 = \tilde{\boldsymbol{\beta}}_t \circ e^{j \tilde{\boldsymbol{\phi}}_t} - \boldsymbol{\theta}_t = \mathbf{0}, \\[-0.5em]
    &\boldsymbol{\mu}_4 = \tilde{\boldsymbol{\beta}}_r \circ e^{j \tilde{\boldsymbol{\phi}}_r} - \boldsymbol{\theta}_r = \mathbf{0},
  \end{align}
\end{subequations}
where $\circ$ denotes the element-wise product, $\mathbf{1}$ is the all-ones vector, and $\mathbf{0}$ is the all-zeros vector. First, we show that the gradients of $\{\boldsymbol{\mu}_i\}_{i=1}^4$ with respect to the vector of variables $\boldsymbol{\omega} = \left[ \tilde{\boldsymbol{\beta}}_t^T, \tilde{\boldsymbol{\beta}}_r^T, \tilde{\boldsymbol{\phi}}_t^T, \tilde{\boldsymbol{\phi}}_r^T, \boldsymbol{\theta}_t^T, \boldsymbol{\theta}_t^T   \right]^T$ are linearly independent. It can be observed that since only $\boldsymbol{\mu}_4$ contains variable $\boldsymbol{\theta}_r$, the gradients of $\{\boldsymbol{\mu}_i\}_{i=1}^4$ are linearly independent if and only if those of $\{\boldsymbol{\mu}_i\}_{i=1}^3$ are linearly independent. Similarly, since only $\boldsymbol{\mu}_3$ contains variable $\boldsymbol{\theta}_t$, the gradients of $\{\boldsymbol{\mu}_i\}_{i=1}^3$ are linearly independent if and only if those of $\{\boldsymbol{\mu}_i\}_{i=1}^2$ are linearly independent. Furthermore, the gradients of $\{\boldsymbol{\mu}_i\}_{i=1}^2$ are linearly independent since $\boldsymbol{\mu}_2$ does not contain $\tilde{\boldsymbol{\beta}}_t$ and $\tilde{\boldsymbol{\beta}}_r$ but $\boldsymbol{\mu}_1$ does. As a consequence, the gradients of $\{\boldsymbol{\mu}_i\}_{i=1}^4$ must be linearly independent. Then, according to the definition of MFCQ \cite[Appendix C]{shi2020penalty2}, we are left to show that there exists a matrix $\mathbf{D}_{\mathbf{W}}$ and a vector $\mathbf{d}_{\boldsymbol{\omega}}$, such that
\begin{subequations}
  \begin{align}
    &\mathrm{Re}\{\mathrm{tr}(\mathbf{W} \mathbf{D}_{\mathbf{W}}^H )\} < 0, \\
    &\nabla \boldsymbol{\mu}_i \mathbf{d}_{\boldsymbol{\omega}} = \mathbf{0}, \forall i \in \{1,2,3,4\},
  \end{align}
\end{subequations}
where $\nabla \boldsymbol{\mu}_i$ denotes the Jacobian matrix of $\boldsymbol{\mu}_i$ with respect to ${\boldsymbol{\omega}}$. It can be readily proved that the above equations are satisfied by setting $\mathbf{D}_{\mathbf{W}} = -\mathbf{W}$ and $\mathbf{d}_{\boldsymbol{\omega}} = \mathbf{0}$, which completes the proof.

\bibliographystyle{IEEEtran}
\bibliography{mybib}

\end{document}